\documentclass[aps,prl,reprint,twocolumn,nofootinbib,floatfix,
  superscriptaddress]{revtex4-2}

\usepackage[T1]{fontenc}
\usepackage{lmodern}
\usepackage{microtype}
\usepackage{float}
\usepackage{amsmath,amssymb,amsthm,mathtools,mathrsfs}
\usepackage{tikz}
\usepackage{capt-of}
\usepackage[bookmarksnumbered=true,bookmarksopen=true]{hyperref}
 \hypersetup{colorlinks,%
             linkcolor=[rgb]{0,0.3,0.6}, %
             citecolor=[rgb]{0,0.3,0.6}, %
             urlcolor=[rgb]{0,0.3,0.6}}

\usetikzlibrary{arrows.meta,patterns}
\newcommand{\Res}{\operatorname*{Res}}

\newcommand{\arsinh}{\operatorname{arsinh}}

\newtheorem{lemma}{Lemma}
\newtheorem{proposition}[lemma]{Proposition}
\begin{document}

\title{Exact All-Level Positivity of the Bosonic Veneziano Amplitude}

\author{Qi Chen}
\email{chenqi25@ucas.ac.cn}
\affiliation{Kavli Institute for Theoretical Sciences, University of Chinese
  Academy of Sciences, Beijing 100190, China}

\author{Yuan Yin}
\email{yinyuan@kias.re.kr}
\affiliation{School of Physics, Korea Institute for Advanced Study,
  Seoul 02455, Korea}

\begin{abstract}
Tree-level partial-wave unitarity of the open bosonic Veneziano amplitude
follows indirectly from the no-ghost theorem, but a direct proof from the
beta function has been missing at the critical dimension.  We prove
nonnegativity at every mass level and spin throughout the sharp range
$3<D\leq26$, using a Bessel recurrence that reduces the infinite sign problem
to a manifestly positive bulk and a finite exact boundary.
\end{abstract}

\maketitle

\paragraph{Introduction.} The Veneziano amplitude is the prototype of a relativistic scattering
amplitude with an infinite spectrum.  In one Euler beta function it realizes
duality between crossed channels, a tower of real simple poles with
polynomial residues, and Regge behavior governed by a linear trajectory
\cite{Veneziano:1968yb,Virasoro:1969me,Shapiro:1970gy,Coon:1969yw,Cheung:2022mkw,Geiser:2022exp,Cheung:2023uwn}.  Its poles describe open-string excitations and its
residues encode their couplings~\cite{Arkani-Hamed:2017jhn,Arkani-Hamed:2023jwn,Arkani-Hamed:2024nzc,Cheung:2025nhw}.  This economy helped reveal the string
interpretation of the dual-resonance model and still makes the amplitude a
benchmark for asking how much consistency is visible directly on shell~\cite{Cheung:2017pzi,Caron-Huot:2016icg,Haring:2023zwu,
Cheung:2023adk,Cheung:2024uhn,Cheung:2024obl,Cheung:2025tbr,
Berman:2023jys,Chiang:2023quf,Berman:2024wyt,Albert:2024yap,
Elvang:2026pmc,Boisvert:2026sfh}.

Unitarity is essential but not manifest in the beta-function formula.  At a
tree-level pole, the amplitude must factorize through on-shell intermediate
states.  After decomposing the residue into $SO(D-1)$ spin partial waves,
each coefficient is, up to a positive normalization, a sum of absolute
squares of three-point couplings and must therefore be nonnegative~\cite{Arkani-Hamed:2017jhn,Arkani-Hamed:2022gsa}.  A
negative coefficient would preclude factorization through a positive-norm
state space.  Since the degrees of the residues grow without bound, this is
an infinite family of inequalities, one for every mass level and allowed
spin.  Remarkably, although the beta function contains no explicit
spacetime dimension, its partial-wave basis does.  The first massive scalar
coefficient is proportional to $26-D$, showing that positivity is impossible
for $D>26$~\cite{Arkani-Hamed:2022gsa}.  This establishes a necessary endpoint, but not its sufficiency:
a higher-level coefficient could in principle impose a stronger bound.

At $D=26$, positivity follows conventionally from the no-ghost theorem
~\cite{DelGiudice:1971yjh,Brower:1972wj,Goddard:1972iy}.  Covariant quantization imposes the Virasoro constraints
on the oscillator Fock space and removes null states, producing a
nonnegative-norm physical state space.  Factorization through this space
then implies positivity of the residues.  This gives a complete state-space
explanation of ghost freedom, but it does not reveal the identity within the
beta function that enforces the sign of each Gegenbauer coefficient.  A
direct amplitude-level proof is therefore complementary: it asks how the
analytic structure of the amplitude itself organizes the infinitely many
unitarity inequalities~\cite{Green:2019tpt,Maity:2021obe}.

The double-contour representation of Ref.~\cite{Arkani-Hamed:2022gsa} made substantial
progress in this direction.  For the open bosonic amplitude it proves
all-level, all-spin positivity in $D\leq10$ and provides partial and
asymptotic information above that range, but not a uniform proof at $D=26$.
Mansfield subsequently completed the analogous type-I problem through its
critical dimension $D=10$ ~\cite{Mansfield:2025gca}.  Coon and hypergeometric
deformations further illustrate the value of amplitude-intrinsic methods
when no conventional string state-space construction is known
~\cite{Chakravarty:2022vrp,Bhardwaj:2022lbz,Figueroa:2022onw,
Jepsen:2023sia,Rigatos:2024beq,Wang:2024wcc,Rigatos:2023asb,
Mansfield:2024wjc,Bhardwaj:2024klc}.

Here we prove directly from the beta function that every bosonic Veneziano
partial-wave coefficient is nonnegative at every level and spin in $D=26$.
A positive Gegenbauer connection formula extends the result to all real
$3<D<26$, while the first massive scalar excludes every $D>26$; the range is
therefore sharp.  The parity zeros and exceptional low-level saturations are
handled exactly.

The proof converts the infinite sign problem into a finite-width one.  An
exact generating-function transform reduces each allowed partial wave to one
Taylor coefficient.  Its Bessel equation yields an acyclic recurrence with a
manifestly positive bulk and a fixed six-step terminal strip, independent of
level and spin.  A first-exit expansion reduces this strip to an exact
integer-polynomial identity whose positive coefficients establish all
remaining signs uniformly.  The accompanying Mathematica notebook verifies
this finite identity using exact arithmetic, without scanning levels or
spins; complete derivations are given in the Supplemental Material.  The result thus provides an amplitude-level mechanism
complementary to the no-ghost theorem and suggests a strategy for other
dual-resonance amplitudes.

\emph{The Veneziano amplitude.} The Veneziano amplitude is given by~\cite{Veneziano:1968yb}:
\begin{equation}
{\cal A}(s,t)=\frac{\Gamma(-s-1)\Gamma(-t-1)}{\Gamma(-s-t-2)}
\end{equation}
At tree level, unitarity is encoded in factorization at physical poles.  Near
such a pole, the residue is a sum over on-shell intermediate states weighted
by products of their couplings.  In a positive-norm partial-wave basis, the
corresponding coefficients must therefore be nonnegative, motivating a direct
analysis of the pole residues.  Define the $s$-channel residue polynomial
$R_n(t)$ at the pole $s=n$ by
\begin{align}
 R_n(t)&=-\Res_{s=n}{\cal A}(s,t)
 =\frac{(t+2)_{n+1}}{(n+1)!},
 \label{eq:residue}
\end{align}
where $n=-1,0,1,\ldots$.  Thus $R_n$ differs by a minus sign
from the complex-analytic residue of the unsigned quotient.  Because
Gegenbauer polynomials furnish the partial-wave basis in $D$ dimensions,
expanding the residues in this basis isolates the contribution of each
intermediate spin and allows its unitarity properties to be examined
separately.  With
$x=\cos\theta$ and $t=(n+4)(x-1)/2$, the residues can be expanded as
\begin{equation}
 R_n(x)=\sum_{j\geq0}B^D_{n,j}
 C_j^{(D-3)/2}(x),\qquad R_{-1}=1.
 \label{eq:partial-waves}
\end{equation}
Here $C_j^{(D-3)/2}$ is the spin-$j$ Gegenbauer partial wave in $D$
dimensions, while $B^D_{n,j}$ is its coefficient at level $n$.  In a
positive-norm basis, factorization gives these coefficients a direct physical
meaning: up to a positive convention-dependent normalization, they are sums
of modulus squares of three-point couplings to intermediate spin-$j$ states,
or equivalently the cross sections for resonant production of those states.
At criticality the no-ghost theorem guarantees this state-space
interpretation; \emph{our proof establishes the same signs directly from the
amplitude, without constructing the intermediate states.}  In this letter, we prove that,
\begin{equation}
 B^D_{n,j}\geq0\quad
 (n\geq-1,\ j\geq0,\ 3<D\leq26).
 \label{eq:theorem}
\end{equation}
At $D=26$, $B^{26}_{n,j}$ vanishes unless $0\leq j\leq n+1$ and
$j\equiv n+1\pmod2$; within this allowed sector, the only zeros are
$B^{26}_{1,0}=B^{26}_{2,1}=0$.  The upper bound $D\leq26$ is sharp, as is
already evident from the scalar channel of the first massive pole:
\begin{equation}
 B^D_{1,0}=\frac{26-D}{8(D-1)}.
 \label{eq:first-obstruction}
\end{equation}
This coefficient vanishes at $D=26$ and is negative for every $D>26$.
Conversely, the all-level theorem shows that no further obstruction arises at
higher poles for $3<D\leq26$.

\emph{Proof outline.}
At $D=26$, an exact generating-function transform expresses each allowed
$B^{26}_{n,j}$ as the product of a positive prefactor and a single Taylor
coefficient $h_\ell$, where $\ell$ is the daughter-trajectory depth.  The Bessel
equation for this generating function yields an acyclic recurrence for its
Taylor coefficients.  After reparametrizing the indices by nonnegative
integers, the recurrence kernel is manifestly positive everywhere except in
a fixed six-step strip adjacent to the target coefficient.  Expanding only
inside this strip and stopping at its first exit into the positive bulk
expresses $h_\ell$ as a sum of coefficients already known to be positive,
with total exit weights $E_R$.  The Bessel equation thus supplies the
infinite-to-finite mechanism; one fixed integer polynomial with strictly
positive coefficients proves $E_R>0$ uniformly in level and spin.  The
shallow cases $\ell<6$ are checked from exact closed forms.  Finally, the
positive Gegenbauer connection formula descends the critical result to every
$3<D<26$, while the first massive scalar coefficient excludes every $D>26$.

\emph{Reduction to one coefficient.}
We first replace the spin-$j$ Gegenbauer projection at level $n$ by a
one-variable coefficient-extraction problem.  This step has two purposes:
it isolates all sign information in a single Taylor coefficient multiplied
by a manifestly positive normalization, and it packages that coefficient
in a generating function to which the Bessel equation can be applied.

To make the parity manifest, define the shifted pole level $m=n+1$.  The residue polynomial then has degree $m$ and definite parity,
\begin{equation}
R_n(-x)=(-1)^mR_n(x).
\end{equation}
Consequently, its Gegenbauer expansion can be nonzero only when
$0\leq j\leq m$ and $j\equiv m\pmod 2$.  For notational convenience, introduce the following combinations:
\begin{equation}
 \begin{aligned}
  \ell&=\frac{m-j}{2},&
  \qquad c&=j+2\ell+1=m+1,\\[-0.2em]
  b&=j+\frac{25}{2},&
  \kappa&=c+2=m+3.
 \end{aligned}
\end{equation}
The leading Regge trajectory has spin $j=m$; within the fixed-parity
sector, each increment of $\ell$ lowers the spin by two.  Thus $\ell$ is
the daughter-trajectory depth probed by the residue.

The Gegenbauer coefficients may be obtained from the standard orthogonality
integral, or equivalently from a contour integral in the complexified scattering-angle variable.  Instead, we avoid this contour representation by
using the central-factorial generating function, which yields
\begin{equation}
 R_n(x)=[u^m]\frac{
 \exp\!\bigl(\kappa x\,\operatorname{arsinh}(u/2)\bigr)}
 {\sqrt{1+u^2/4}}.
\end{equation}
Here $[u^m]F(u)$ denotes the coefficient of $u^m$ in the Taylor expansion of
$F(u)$ about $u=0$.  The Gegenbauer projection is implemented through the
plane-wave expansion
\begin{equation}
 e^{ax}=\sum_{j\geq0}
 \frac{\Gamma(\lambda)a^j}{2^j\Gamma(j+\lambda)}
 {}_0F_1\!\left(;j+\lambda+1;\frac{a^2}{4}\right)
 C_j^\lambda(x),
 \label{eq:gegenbauer-plane-wave}
\end{equation}
valid for $\lambda>0$.  Thus each Gegenbauer coefficient is expressed in
terms of ${}_0F_1$, equivalently a modified Bessel function.  Specializing to
$\lambda=23/2$, define the coefficients $h_p$ by the even Taylor expansion
\begin{equation}
\left(\frac{z}{\sinh z}\right)^c{}_0F_1\!\left(;b;\frac{\kappa^2z^2}{4}\right)=\sum_{p\geq0}h_pz^{2p}.
\label{eq:bessel-reduction}
\end{equation}
Combining the plane-wave expansion with the central-factorial generating
function then gives
\begin{equation}
 B^{26}_{n,j}={\cal P}_{n,j}h_\ell,
 \label{eq:coefficient-reduction}
\end{equation}
where the prefactor is
\begin{equation}
 {\cal P}_{n,j}=2^{-m-j}
 \frac{\Gamma(23/2)\kappa^j}{\Gamma(j+23/2)}>0.
 \label{eq:positive-prefactor}
\end{equation}
The definition of $h_p$ gives $h_0=1$.
For every allowed pair $(n,j)$,
\begin{equation}
 B^{26}_{n,j}\geq0
 \quad\Longleftrightarrow\quad
 h_\ell\geq0,
 \qquad
 \ell=\frac{n+1-j}{2} \geq 0.
\end{equation}
The critical-dimensional theorem is therefore reduced to proving
$h_\ell\geq0$ for all $\ell,j \geq 0$.

\emph{Bessel recurrence and positive bulk.} The generating-function
representation reduces the positivity problem to the analysis of its Taylor
coefficients.  Applying the differential equation satisfied by ${}_0F_1$ in
Eq.~\eqref{eq:bessel-reduction} and conjugating by $(z/\sinh z)^c$ yields the acyclic recurrence
\begin{equation}
 2p(2j+2p+23)h_p
 =\sum_{r=2}^{p+1}\frac{2^{2r-1}}{(2r)!}M_{p,r}h_{p+1-r}.
 \label{eq:recurrence}
\end{equation}
The positivity analysis therefore reduces to determining the signs of the
recurrence coefficients $M_{p,r}$, given by
\begin{equation}
\begin{split}
M_{p,r}=&c^2+4(r-1)c\ell+(4r^2-24r+21)c+2r(2r-1)\\
&-2(p+1-r)(2p+23-2r+2cr-4\ell).
\end{split}
\end{equation}
For $p\leq\ell-6$, reparametrize the allowed domain using the nonnegative
slack variables $X=\ell-p-6$, $Y=r-2$, and $Q=p-r+1$. This change of variables maps the allowed parameter domain to the nonnegative orthant.  Direct substitution expresses $M_{p,r}$ as a polynomial in $j,X,Y,Q$ with 20 strictly positive coefficients; the explicit polynomial is given in the SM.  Hence $M_{p,r}$ is manifestly positive throughout the bulk region shown in Fig.~\ref{fig:bulkcone}.  Since $h_0=1$, Eq.~\eqref{eq:recurrence} therefore proves
\begin{equation}
 h_p>0\qquad(0\leq p\leq\ell-6).
 \label{eq:bulk}
\end{equation}
Since ${\cal P}_{n,j}>0$, Eq.~\eqref{eq:coefficient-reduction} identifies
the sign of $B_{n,j}^{26}$ with that of the corresponding coefficient
$h_\ell$.  The remaining nonmanifest region consists of exactly six boundary coefficients, with a width independent of $n$ and $j$.  Thus, for
$\ell\geq6$, the termwise argument proves $h_p>0$ for every
$0\leq p\leq\ell-6$, leaving only
$h_{\ell-5},\ldots,h_\ell$ to be analyzed.
\begin{figure}[t]
\centering
\resizebox{0.98\columnwidth}{!}{%
\begin{tikzpicture}[x=0.55cm,y=0.34cm,font=\small]
  \fill[green!9] (0,0) rectangle (10.2,9.8);
  \fill[gray!11] (10.2,0) rectangle (12.3,9.8);
  \fill[pattern=north east lines,pattern color=gray!38]
        (10.2,0) rectangle (12.3,9.8);

  \foreach \x in {0.55,1.45,...,9.55}{
    \foreach \y in {0.55,1.55,...,9.55}{
      \fill[black!62] (\x,\y) circle[radius=0.82pt];
    }
  }

  \foreach \x in {10.38,10.73,11.08,11.43,11.78,12.13}{
    \foreach \y in {0.55,1.55,...,9.55}{
      \fill[black!62] (\x,\y) circle[radius=0.82pt];
    }
  }

  \draw[black!52,line width=0.45pt]
        (0,0) rectangle (12.3,9.8);
  \draw[dashed,green!42!black,line width=0.75pt]
        (10.2,0) -- (10.2,9.8);

  \node[
    align=center,
    fill=green!9,
    inner sep=3pt
  ] at (5.1,5.15)
    {\textbf{positive bulk}\\[-1pt]
     $p\leq\ell-6$};

  \node[
    align=center,
    anchor=south,
    font=\scriptsize
  ] at (11.25,10.55)
    {\textbf{six-step boundary strip}\\[-1pt]
     $p=\ell-5,\ldots,\ell$};

  \draw[black!45,line width=0.4pt]
        (11.25,10.45) -- (11.25,9.9);

  \node[rotate=90,text=black!75]
        at (-0.72,4.9) {spin $j$};
  \node[text=black!75]
        at (6.15,-0.82) {recurrence index $p$};
\end{tikzpicture}%
}
\caption{Schematic recurrence domain at fixed $\ell$.  The coefficients are
manifestly positive for $p\leq\ell-6$; only the six values
$p=\ell-5,\ldots,\ell$ lie in the nonmanifest boundary strip.  The
horizontal scale is compressed in the strip.}
\label{fig:bulkcone}
\end{figure}
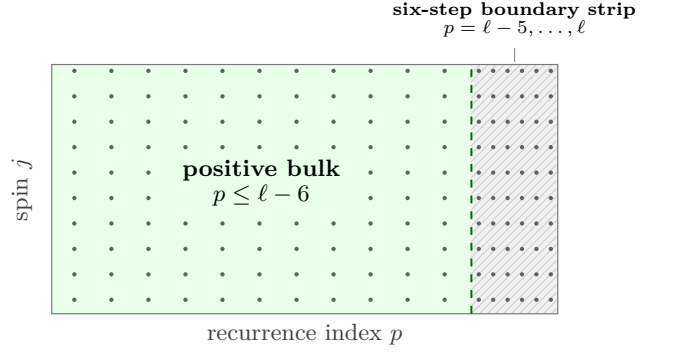

\emph{Finite boundary.}
Let $d=\ell-p$ denote the depth of the coefficient
$h_p=h_{\ell-d}$ relative to the target $h_\ell$.  By
Eq.~\eqref{eq:bulk}, $h_{\ell-d}>0$ for every $6\leq d\leq\ell$.
It therefore remains to control the fixed boundary strip
$0\leq d\leq5$, within which the individual recurrence coefficients
need not have a definite sign.  Starting from the recurrence for
$h_\ell$, we recursively substitute only those terms that remain in
this strip and terminate each branch upon its first entry into the
positive bulk, at some depth $R\geq6$.  To formulate the expansion precisely, normalize a
recurrence edge by
\begin{equation}
 C_{p,r}=\frac{\alpha_rM_{p,r}}
 {2p(2j+2p+23)}.
 \label{eq:normalized-edge}
\end{equation}
For $\ell\geq6$, let $F_d$ be the total algebraic weight of paths
that reach depth $d$ without leaving the strip, and let $E_R$ be
the total weight of paths whose first strip-exit is to depth $R$.
Explicitly,
\begin{align}
 F_0&=1,\qquad
 F_d=\sum_{u=0}^{d-1}
 F_uC_{\ell-u,d-u+1},
 \quad d=1,\ldots,5,\notag\\
 E_R&=\sum_{u=0}^{5}
 F_uC_{\ell-u,R-u+1},
 \quad R=6,\ldots,\ell.
 \label{eq:path-weights}
\end{align}
Every recurrence step strictly increases $d$, so this expansion is
finite and exhaustive.  Summing over all possible first exits gives
the exact identity
\begin{equation}
 h_\ell=\sum_{R=6}^{\ell}E_Rh_{\ell-R}.
 \label{eq:first-exit}
\end{equation}
Each landing coefficient $h_{\ell-R}$ belongs to the positive bulk.
It therefore remains only to prove $E_R>0$ for
$6\leq R\leq\ell$.  The $F_d$ are intermediate algebraic weights
and need not be positive.
Table~\ref{tab:first-exit} displays the summation process in
Eqs.~\eqref{eq:path-weights} and \eqref{eq:first-exit}.

\begin{table*}[t]
\caption{\label{tab:first-exit}
Bookkeeping for the summation in Eq.~\eqref{eq:first-exit}. Read the table from top to bottom.  The cell
$(h_{\ell-a},h_{\ell-d})$, with $d>a$, is the contribution
$F_aC_{\ell-a,d-a+1}$ to the coefficient of $h_{\ell-d}$ obtained by
replacing $h_{\ell-a}$ with its recurrence relation.  The bottom entry of
each column is the resulting total coefficient $F_d$ of $h_{\ell-d}$, or
the exit weight $E_R$ in the final column.}
\centering
\begingroup
\small
\setlength{\tabcolsep}{3.2pt}
\renewcommand{\arraystretch}{1.15}
\begin{ruledtabular}
\begin{tabular}{l c c c c c c}
\multicolumn{1}{c}{
  \shortstack{coefficient already\\reached}
}
&
\multicolumn{5}{c}{
  one-step target inside the boundary strip
}
&
\multicolumn{1}{c}{
  \shortstack{first exit\\(fixed $R$)}
}
\\[-0.15em]
&
$h_{\ell-1}$
&
$h_{\ell-2}$
&
$h_{\ell-3}$
&
$h_{\ell-4}$
&
$h_{\ell-5}$
&
$h_{\ell-R}$
\\
\colrule
$h_\ell\ (a=0)$
&
$F_0C_{\ell,2}$
&
$F_0C_{\ell,3}$
&
$F_0C_{\ell,4}$
&
$F_0C_{\ell,5}$
&
$F_0C_{\ell,6}$
&
$F_0C_{\ell,R+1}$
\\
$h_{\ell-1}\ (a=1)$
&
\textemdash
&
$F_1C_{\ell-1,2}$
&
$F_1C_{\ell-1,3}$
&
$F_1C_{\ell-1,4}$
&
$F_1C_{\ell-1,5}$
&
$F_1C_{\ell-1,R}$
\\
$h_{\ell-2}\ (a=2)$
&
\textemdash
&
\textemdash
&
$F_2C_{\ell-2,2}$
&
$F_2C_{\ell-2,3}$
&
$F_2C_{\ell-2,4}$
&
$F_2C_{\ell-2,R-1}$
\\
$h_{\ell-3}\ (a=3)$
&
\textemdash
&
\textemdash
&
\textemdash
&
$F_3C_{\ell-3,2}$
&
$F_3C_{\ell-3,3}$
&
$F_3C_{\ell-3,R-2}$
\\
$h_{\ell-4}\ (a=4)$
&
\textemdash
&
\textemdash
&
\textemdash
&
\textemdash
&
$F_4C_{\ell-4,2}$
&
$F_4C_{\ell-4,R-3}$
\\
$h_{\ell-5}\ (a=5)$
&
\textemdash
&
\textemdash
&
\textemdash
&
\textemdash
&
\textemdash
&
$F_5C_{\ell-5,R-4}$
\\
\colrule
\multicolumn{1}{c}{column sum}
&
$F_1$
&
$F_2$
&
$F_3$
&
$F_4$
&
$F_5$
&
$E_R$
\\
\colrule
\multicolumn{7}{c}{
$\displaystyle
\text{all first exits:}\qquad
h_\ell=\sum_{R=6}^{\ell}E_Rh_{\ell-R}$
}
\end{tabular}
\end{ruledtabular}
\endgroup
\end{table*}

Using the normalized recurrence edge in
Eq.~\eqref{eq:normalized-edge} and the internal path weights $F_d$
defined in Eq.~\eqref{eq:path-weights}, the exit weights $E_R$ take
the explicit form
\begin{align}
 E_R&=\frac{\alpha_{R+1}\,
 \mathscr P(j,R-6,\ell-R)}{{\cal D}_{\ell,j}},
 \label{eq:certificate}\\
 {\cal D}_{\ell,j}
 &=958003200\prod_{u=0}^{5}(\ell-u)\notag\\[-0.2em]
 &\quad\times\prod_{u=0}^{5}(2j+2\ell+23-2u).
\end{align}
Here $\mathscr P\in\mathbb Z[J,X,Y]$ has multidegree
$(12,22,11)$, total degree 23, and exactly 1360 nonzero monomials.
Every coefficient is a positive integer.

Thus $\mathscr P$ is strictly positive on the closed nonnegative orthant,
all denominator factors in Eq.~\eqref{eq:certificate} are positive, and
$E_R>0$.  Equations \eqref{eq:bulk} and \eqref{eq:first-exit} imply
$h_\ell>0$ for every $\ell\geq6$.  The accompanying Mathematica notebook reconstructs the explicit
polynomial $\mathscr P$ directly from
Eqs.~\eqref{eq:recurrence} and \eqref{eq:path-weights}, and verifies that every coefficient is a
positive integer.

For the six low trajectories, the positive normalization
$L_\ell=\ell!(j+25/2)_\ell h_\ell$ gives
\begin{equation}
 \begin{aligned}
 L_0&=1,\qquad L_1=\frac{j(j-1)}{12},\\[-0.15em]
 720L_2&=5j^4+38j^3+23j^2\\[-0.15em]
 &\quad-210j+720.
 \end{aligned}
 \label{eq:low-ell}
\end{equation}
The remaining quadratic $23j^2-210j+720$ has negative discriminant, so $L_2>0$ for
$j\geq0$; $L_3$ has a positive falling-factorial expansion and $L_4,L_5$
have positive monomial expansions, given exactly in the SM.  Hence the only
extra zeros are those of $L_1$, namely $(n,j)=(1,0),(2,1)$.

\emph{Dimension descent and sharpness.}
For $0<\mu\leq\nu$, the connection formula
\begin{align}
 C_J^\nu(x)&=\sum_{q=0}^{\lfloor J/2\rfloor}
 d_{J,q}(\nu,\mu)C_{J-2q}^\mu(x),\notag\\[-0.2em]
 d_{J,q}&=\frac{(\nu-\mu)_q(\nu)_{J-q}(J-2q+\mu)}
 {\mu\,q!(\mu+1)_{J-q}}\geq0
 \label{eq:connection}
\end{align}
expresses every $D=26$ partial wave as a positive combination of those at
$\mu=(D-3)/2$.  This proves Eq.~\eqref{eq:theorem} for every real
$3<D\leq26$, in particular every integer spacetime dimension
$4\leq D\leq26$.  Conversely,
\begin{equation}
 \begin{aligned}
 R_1(x)&=\frac{25x^2-1}{8}\\[-0.25em]
 &=\frac{25C_2^{(D-3)/2}(x)}{4(D-3)(D-1)}
 +\frac{26-D}{8(D-1)}C_0(x).
 \end{aligned}
 \label{eq:sharpness}
\end{equation}
Therefore $B^D_{1,0}<0$ for every $D>26$.

\paragraph{Discussion.} The amplitude thus identifies the critical dimension in two complementary
ways.  The first massive scalar shows that positivity must fail for $D>26$,
whereas the Bessel recurrence proves uniformly that no higher pole or deeper
daughter trajectory strengthens this bound at $D=26$.  This result
complements, rather than replaces, the no-ghost theorem: the latter
establishes positivity by constructing the physical state space, while our
argument exhibits how the same conclusion is encoded in the analytic
structure of the beta function.

The finite polynomial certificate rigorously closes the recurrence, but
should not yet be regarded as a physical explanation of the remaining
boundary positivity.  A more transparent contour representation or
positive-operator formulation may connect this boundary directly to
factorization and physical-state norms.  Independently, the proof suggests a
concrete test for other dual-resonance amplitudes: derive a holonomic
recurrence, determine whether its nonmanifest region has uniformly bounded
width, and seek a positive representation of the resulting exit weights.
Applying this criterion to Coon and hypergeometric deformations
\cite{Figueroa:2022onw,Rigatos:2023asb} may clarify which ingredients are universal and which
are special to the critical bosonic amplitude.

\paragraph*{\bf Acknowledgements.} We thank Nima Arkani-Hamed and Bo Wang for useful discussions. This work was carried out with assistance from GPT-5.6 Sol. The work of Q.C. is supported by the NSFC Grant No. 12275273. The work of Y.Y. is supported by the KIAS New Generation Research Program under Grant No. PG105801. 

\bibliography{Veneziano}

\clearpage
\onecolumngrid

\setcounter{equation}{0}
\renewcommand{\theequation}{S\arabic{equation}}
\renewcommand{\theHequation}{supplement.\arabic{equation}}

\setcounter{section}{0}
\renewcommand{\thesection}{S\Roman{section}}
\renewcommand{\theHsection}{supplement.\arabic{section}}

\setcounter{lemma}{0}
\allowdisplaybreaks
\setlength{\emergencystretch}{2em}

\begin{center}
{\large\bfseries
Supplemental Material}

\end{center}

\vspace{1em}

This Supplemental Material supplies detailed derivations used in the
Letter.  The first two sections fix the residue sign and derive the exact
Gegenbauer--Bessel reduction.  The next two derive the triangular recurrence
and its positive bulk.  The first-exit construction and finite exact
certificate are then proved, followed by the six exceptional trajectories,
dimension descent, and sharpness.

\section{Physical residue and central factorials}
\label{sec:residue}

Consider the unsigned quotient
\begin{equation}
 {\cal A}(s,t)=
 \frac{\Gamma(-s-1)\Gamma(-t-1)}{\Gamma(-s-t-2)}.
 \label{eqS:amplitude}
\end{equation}
At $s=n$, where $n=-1,0,1,\ldots$, set $m=n+1\geq0$.  Since
\begin{equation}
 \Res_{s=n}\Gamma(-s-1)=\frac{(-1)^{m+1}}{m!},
 \qquad
 \frac{\Gamma(-t-1)}{\Gamma(-t-m-1)}=(-1)^m(t+2)_m,
\end{equation}
the complex-analytic residue is
\begin{equation}
 \Res_{s=n}{\cal A}(s,t)=-\frac{(t+2)_m}{m!}.
\end{equation}
The polynomial that multiplies the physical propagator $(n-s)^{-1}$ is
therefore
\begin{equation}
 R_n(t):=-\Res_{s=n}{\cal A}(s,t)=\frac{(t+2)_m}{m!}.
 \label{eqS:physical-residue}
\end{equation}
This sign convention is used throughout.

At the pole, write
\begin{equation}
 t=\frac{n+4}{2}(x-1),\qquad x=\cos\theta,
 \qquad \kappa=m+3=n+4.
\end{equation}
Define the central-factorial polynomial
\begin{equation}
 Q_m(y)=\prod_{a=-(m-1)/2}^{(m-1)/2}(y-a),
 \label{eqS:central-factorial}
\end{equation}
where factors advance in unit steps and $Q_0=1$.  Equation
\eqref{eqS:physical-residue} becomes
\begin{equation}
 R_n(x)=\frac1{m!}Q_m\!\left(\frac{\kappa x}{2}\right).
 \label{eqS:central-residue}
\end{equation}
Consequently $R_n(-x)=(-1)^mR_n(x)$, which already gives the parity
selection rule.

\begin{lemma}[Central-factorial generating function]
For every complex $y$,
\begin{equation}
 \sum_{m\geq0}\frac{Q_m(y)}{m!}u^m
 =\frac{\exp\!\bigl(2y\arsinh(u/2)\bigr)}{\sqrt{1+u^2/4}}
 \label{eqS:central-gf}
\end{equation}
as a formal series at $u=0$.
\end{lemma}

\begin{proof}
Set $u=2\sinh z$ and then $w=e^{2z}$.  Coefficient extraction gives
\begin{align}
 [u^m]\frac{e^{2y\arsinh(u/2)}}{\sqrt{1+u^2/4}}
 &=2^{-m}\Res_{z=0}\frac{e^{2yz}}{\sinh^{m+1}z}\,dz\notag\\
 &=\Res_{w=1}\frac{w^{y+(m-1)/2}}{(w-1)^{m+1}}\,dw\notag\\
 &=\binom{y+(m-1)/2}{m}=\frac{Q_m(y)}{m!}.
\end{align}
The identity follows coefficient by coefficient.
\end{proof}

Combining Eqs.~\eqref{eqS:central-residue} and \eqref{eqS:central-gf},
\begin{equation}
 R_n(x)=[u^m]\frac{
 \exp\!\bigl(\kappa x\arsinh(u/2)\bigr)}{\sqrt{1+u^2/4}}.
 \label{eqS:residue-u}
\end{equation}

\section{Exact Gegenbauer--Bessel reduction}
\label{sec:bessel}

For $\lambda>0$, the branch-free Gegenbauer plane-wave identity is
\begin{equation}
 e^{ax}=\sum_{j\geq0}
 \frac{\Gamma(\lambda)}{2^j\Gamma(j+\lambda)}a^j
 {}_0F_1\!\left(;j+\lambda+1;\frac{a^2}{4}\right)
 C_j^\lambda(x).
 \label{eqS:plane-wave}
\end{equation}
It follows either by Gegenbauer projection or by inserting the series of
$I_{j+\lambda}(a)$ into
\begin{equation}
 e^{ax}=2^\lambda\Gamma(\lambda)
 \sum_{j\geq0}(j+\lambda)a^{-\lambda}I_{j+\lambda}(a)C_j^\lambda(x).
\end{equation}

At $D=26$, $\lambda=23/2$.  Put $z=\arsinh(u/2)$ in
Eq.~\eqref{eqS:residue-u}, insert Eq.~\eqref{eqS:plane-wave}, and project
onto $C_j^{23/2}$.  This gives
\begin{equation}
 B^{26}_{n,j}=
 \frac{\Gamma(23/2)\kappa^j}{2^j\Gamma(j+23/2)}
 [u^m]\frac{z^j}{\cosh z}
 {}_0F_1\!\left(;j+\frac{25}{2};\frac{\kappa^2z^2}{4}\right).
 \label{eqS:B-u}
\end{equation}
In the coefficient contour, $u=2\sinh z$ and
$du=2\cosh z\,dz$, so the Jacobian cancels $1/\cosh z$.  Hence
\begin{equation}
 B^{26}_{n,j}=2^{-m-j}
 \frac{\Gamma(23/2)\kappa^j}{\Gamma(j+23/2)}
 [z^{m-j}]\left(\frac z{\sinh z}\right)^{m+1}
 {}_0F_1\!\left(;j+\frac{25}{2};\frac{\kappa^2z^2}{4}\right).
 \label{eqS:B-z}
\end{equation}

Both $Q_m$ and $C_j^{23/2}$ have the parity of their degree.  Thus the
coefficient vanishes if $j>m$ or $j\not\equiv m\pmod2$.  In the allowed
sector define
\begin{equation}
 \ell=\frac{m-j}{2}\in\mathbb N_0,\qquad
 c=m+1=j+2\ell+1,\qquad b=j+\frac{25}{2},\qquad \kappa=c+2.
 \label{eqS:parameters}
\end{equation}
Let
\begin{equation}
 H_{\ell,j}(z)=
 \left(\frac z{\sinh z}\right)^c
 {}_0F_1\!\left(;b;\frac{\kappa^2z^2}{4}\right)
 =\sum_{p\geq0}h_pz^{2p}.
 \label{eqS:H}
\end{equation}
Then $h_0=1$ and
\begin{equation}
 \boxed{
 B^{26}_{n,j}=2^{-m-j}
 \frac{\Gamma(23/2)\kappa^j}{\Gamma(j+23/2)}h_\ell.}
 \label{eqS:B-exact}
\end{equation}
The prefactor is strictly positive.  As checks,
\begin{equation}
 B^{26}_{-1,0}=1,\quad B^{26}_{0,1}=\frac2{23},\quad
 B^{26}_{1,2}=\frac1{92},\quad B^{26}_{2,3}=\frac1{575}.
\end{equation}

\section{Conjugated Bessel equation and recurrence}
\label{sec:recurrence}

Write
\begin{equation}
 \Phi(z)={}_0F_1\!\left(;b;\frac{\kappa^2z^2}{4}\right).
\end{equation}
Its hypergeometric equation is
\begin{equation}
 \Phi''+\frac{2b-1}{z}\Phi'-\kappa^2\Phi=0.
 \label{eqS:Phi-ode}
\end{equation}
Since $\Phi=(\sinh z/z)^cH$, introduce
$\varphi(z)=\coth z-z^{-1}$.  Conjugating Eq.~\eqref{eqS:Phi-ode} gives
\begin{equation}
 H''+\left(2c\varphi+\frac{2b-1}{z}\right)H'=K(z)H,
 \label{eqS:H-ode}
\end{equation}
where
\begin{equation}
 K(z)=\kappa^2-c\varphi'-c^2\varphi^2
 -(2b-1)c\frac{\varphi}{z}.
\end{equation}
Multiplication by $z^2\sinh^2z$ yields
\begin{align}
 z^2\sinh^2z\,H''
 &+\bigl(2cz^2\sinh z\cosh z
 +(22-4\ell)z\sinh^2z\bigr)H'=N(z)H,
 \label{eqS:multiplied-ode}\\
 N(z)&=\kappa^2z^2\sinh^2z-c^2z^2\cosh^2z+cz^2\notag\\
 &\quad+c(2c-2b+1)z\sinh z\cosh z
 +c(2b-c-2)\sinh^2z.
 \label{eqS:N}
\end{align}

Set
\begin{equation}
 \alpha_r=\frac{2^{2r-1}}{(2r)!}.
\end{equation}
The elementary expansions
\begin{equation}
 \sinh^2z=\sum_{r\geq1}\alpha_rz^{2r},\qquad
 z\sinh z\cosh z=\sum_{r\geq1}r\alpha_rz^{2r},\qquad
 \alpha_{r-1}=\frac{r(2r-1)}2\alpha_r
\end{equation}
show that the coefficient of $z^2$ in $N$ cancels and
\begin{align}
 N(z)&=\sum_{r\geq2}\alpha_r{\cal K}_rz^{2r},\notag\\
 {\cal K}_r&=c^2+4(r-1)c\ell
 +(4r^2-24r+21)c+2r(2r-1).
 \label{eqS:Kr}
\end{align}
Extracting $[z^{2(p+1)}]$ from Eq.~\eqref{eqS:multiplied-ode} gives the
following recurrence.

\begin{proposition}[Triangular recurrence]
For every $p\geq1$,
\begin{equation}
 2p(2j+2p+23)h_p
 =\sum_{r=2}^{p+1}\alpha_rM_{p,r}h_{p+1-r},
 \label{eqS:recurrence}
\end{equation}
where, with $k=p+1-r$,
\begin{equation}
 M_{p,r}={\cal K}_r-2k(2k-1+2cr+22-4\ell).
 \label{eqS:M-compact}
\end{equation}
Equivalently, setting $\delta=\ell-p$,
\begin{align}
 M_{p,r}={}&(8r-4)\delta^2+(16r^2-60r+50)\delta+j^2\notag\\
 &+(4\delta r+8r^2-28r+23)j\notag\\
 &+(8\delta r+16r^2-56r)p+8r^2+20r-24.
 \label{eqS:M-expanded}
\end{align}
\end{proposition}

\begin{proof}
In the summand indexed by $r$, the index of $h$ is $k=p+1-r$.
The $r=1$, $k=p$ derivative contribution is
\begin{equation}
 2p\bigl((2p-1)+2c+22-4\ell\bigr)h_p
 =2p(2j+2p+23)h_p.
\end{equation}
Moving every term with $r\geq2$ to the right gives
Eq.~\eqref{eqS:M-compact}; substituting $c=j+2\ell+1$ and
$\delta=\ell-p$ gives Eq.~\eqref{eqS:M-expanded}.
\end{proof}

\section{Manifestly positive bulk}
\label{sec:bulk}

Suppose $p\leq\ell-6$.  For every term in
Eq.~\eqref{eqS:recurrence}, write
\begin{equation}
 \delta=6+X,\qquad r=2+Y,\qquad p=r-1+Q,\qquad X,Y,Q\geq0.
 \label{eqS:bulk-shift}
\end{equation}
Substitution in Eq.~\eqref{eqS:M-expanded} gives the complete polynomial
\begin{align}
M_{p,r}={}&j^2+4jXY+8jX+8jY^2+28jY+47j\notag\\
&+8X^2Y+12X^2+24XY^2+8XYQ+124XY+16XQ+154X\notag\\
&+16Y^3+16Y^2Q+176Y^2+56YQ+468Y+48Q+492.
\label{eqS:M-positive}
\end{align}
All 20 coefficients are positive.  Since $\alpha_r>0$ and the factor on
the left of Eq.~\eqref{eqS:recurrence} is positive, induction from
$h_0=1$ proves
\begin{equation}
 h_p>0\qquad(0\leq p\leq\ell-6).
 \label{eqS:bulk-positive}
\end{equation}

\section{First-exit identity and exact boundary certificate}
\label{sec:boundary}

For $p\geq1$ define the edge weight
\begin{equation}
 C_{p,r}=\frac{\alpha_rM_{p,r}}{2p(2j+2p+23)}.
 \label{eqS:C}
\end{equation}
For $0\leq d\leq5$, set
\begin{equation}
 F_0=1,\qquad
 F_d=\sum_{u=0}^{d-1}F_uC_{\ell-u,d-u+1}\quad(1\leq d\leq5).
 \label{eqS:F}
\end{equation}
Thus $F_d$ is the sum of the products of edge weights over all recurrence
paths from $h_\ell$ to $h_{\ell-d}$ that remain at cumulative decrement
less than six.  Individual weights and the totals $F_d$ need not be
positive.  For $6\leq R\leq\ell$, define
\begin{equation}
 E_R=\sum_{u=0}^{5}F_uC_{\ell-u,R-u+1}.
 \label{eqS:E}
\end{equation}

\begin{lemma}[First-exit identity]
For $\ell\geq6$,
\begin{equation}
 h_\ell=\sum_{R=6}^{\ell}E_Rh_{\ell-R}.
 \label{eqS:first-exit}
\end{equation}
\end{lemma}

\begin{proof}
Every edge in Eq.~\eqref{eqS:recurrence} strictly lowers the index of $h$,
so the substitution graph is finite and acyclic.  Repeatedly substitute
only while the cumulative decrement is less than six, and stop a path on
its first exit.  Equation~\eqref{eqS:F} sums every internal path to depth
$u$, while Eq.~\eqref{eqS:E} appends every possible first-exit edge to
depth $R$.  The classes labeled by $R=6,\ldots,\ell$ are disjoint and
exhaust all paths, proving Eq.~\eqref{eqS:first-exit}.
\end{proof}

To determine the signs of $E_R$, use
\begin{equation}
 \frac{\alpha_{R-u+1}}{\alpha_{R+1}}
 =4^{-u}\prod_{v=0}^{2u-1}(2R+2-v).
 \label{eqS:alpha-ratio}
\end{equation}
Exact rational simplification of the finite sums
\eqref{eqS:F}--\eqref{eqS:E} then gives
\begin{equation}
 E_R=\frac{\alpha_{R+1}\,
 \mathscr P(j,R-6,\ell-R)}{
 958003200\prod_{u=0}^{5}(\ell-u)
 \prod_{u=0}^{5}(2j+2\ell+23-2u)}.
 \label{eqS:E-certificate}
\end{equation}
This equation may equivalently be taken as the definition of
$\mathscr P(J,X,Y)$: construct the left side from
Eqs.~\eqref{eqS:C}--\eqref{eqS:E}, clear the displayed denominator, and
make the bijective change of variables
\begin{equation}
 j=J,\qquad R=X+6,\qquad \ell=X+Y+6.
 \label{eqS:orthant}
\end{equation}
Thus $j\geq0$ and $6\leq R\leq\ell$ correspond exactly to
$J,X,Y\geq0$.

\begin{proposition}[Finite boundary certificate]
The polynomial defined above satisfies
\begin{align}
 &\mathscr P\in\mathbb Z[J,X,Y],\qquad
 (\deg_J,\deg_X,\deg_Y,\deg_{\mathrm{tot}})=(12,22,11,23),\notag\\
 &|\operatorname{supp}\mathscr P|=1360,\qquad
 [J^aX^bY^c]\mathscr P>0
 \quad\text{for every }(a,b,c)\in\operatorname{supp}\mathscr P,
 \label{eqS:P-data}
\end{align}
with minimum coefficient $16384$ and constant coefficient
\begin{equation}
 [J^0X^0Y^0]\mathscr P
 =1319504646933817344000000.
 \label{eqS:P-constant}
\end{equation}
Consequently $E_R>0$ for every $j\geq0$ and $6\leq R\leq\ell$.
\end{proposition}

\begin{proof}
The algebraic assertions in Eqs.~\eqref{eqS:P-data} and
\eqref{eqS:P-constant} are verified by the standalone exact-arithmetic
program described below.  They are a finite proof object: no value of
$j,R,$ or $\ell$ is scanned.  Coefficientwise positivity together with
the nonzero constant gives, on the entire closed orthant,
\begin{equation}
 \mathscr P(J,X,Y)\geq
 1319504646933817344000000>0.
\end{equation}
This explicit constant is needed at the corner $J=X=Y=0$, corresponding
to $(j,R,\ell)=(0,6,6)$.  Every factor in the denominator of
Eq.~\eqref{eqS:E-certificate} is strictly positive on the stated domain,
so $E_R>0$.
\end{proof}

Combining Eq.~\eqref{eqS:first-exit}, the proposition, and
Eq.~\eqref{eqS:bulk-positive} proves $h_\ell>0$ for every $\ell\geq6$.
The endpoint is included: $h_6=E_6h_0>0$.

\section{The six exceptional trajectories}
\label{sec:low}

For $0\leq\ell\leq5$, define the positive normalization
\begin{equation}
 L_\ell(j)=\ell!\left(j+\frac{25}{2}\right)_\ell h_\ell.
 \label{eqS:L}
\end{equation}
Direct exact expansion of Eq.~\eqref{eqS:H} gives
\begin{equation}
 L_0=1,\qquad L_1=\frac{j(j-1)}{12},
 \label{eqS:L01}
\end{equation}
and
\begin{equation}
 L_2=\frac{5j^4+38j^3+23j^2-210j+720}{720}.
 \label{eqS:L2}
\end{equation}
The last three terms in the numerator of $L_2$ form a quadratic with
discriminant $-22140$ and positive leading coefficient.  Together with
$5j^4+38j^3\geq0$, this proves $L_2>0$ for $j\geq0$.

Write $j^{\underline r}=j(j-1)\cdots(j-r+1)$ and
$j^{\underline0}=1$.  For integer $j\geq0$, every such falling factorial
is nonnegative.  The next trajectory is
\begin{align}
12^3L_3={}&\frac{139968}{5}
+\frac{98496}{35}j^{\underline1}
+\frac{136296}{35}j^{\underline2}
+\frac{19608}{7}j^{\underline3}\notag\\
&+\frac{19332}{35}j^{\underline4}
+\frac{204}{5}j^{\underline5}
+j^{\underline6},
\label{eqS:L3}
\end{align}
which is strictly positive.  Finally,
\begin{align}
3628800L_4={}&175j^8+9380j^7+204094j^6+2264120j^5
+12964735j^4\notag\\
&+33558740j^3+53851476j^2+460077840j+1746230400,
\label{eqS:L4}
\end{align}
and
\begin{align}
95800320L_5={}&385j^{10}+35035j^9+1379774j^8+30585478j^7
+414681881j^6\notag\\
&+3490627963j^5+17801852152j^4+56139102372j^3\notag\\
&+183464019504j^2+915513553152j+2373120820224.
\label{eqS:L5}
\end{align}
All coefficients in Eqs.~\eqref{eqS:L4} and \eqref{eqS:L5} are positive.
Since the normalization in Eq.~\eqref{eqS:L} is positive, $h_\ell\geq0$
for $0\leq\ell\leq5$.  The only zeros are $L_1(0)=L_1(1)=0$, namely
\begin{equation}
 (\ell,j)=(1,0),(1,1)
 \quad\Longleftrightarrow\quad
 (n,j)=(1,0),(2,1).
\end{equation}
Together with the boundary argument, this proves all $D=26$ coefficients,
including the tachyon $B^{26}_{-1,0}=1$.

\section{Positive dimension descent and sharpness}
\label{sec:descent}

\begin{lemma}[Positive Gegenbauer connection]
If $0<\mu\leq\nu$, then
\begin{equation}
 C_J^\nu(x)=\sum_{q=0}^{\lfloor J/2\rfloor}
 d_{J,q}(\nu,\mu)C_{J-2q}^\mu(x),
 \label{eqS:connection}
\end{equation}
where
\begin{equation}
 d_{J,q}(\nu,\mu)=
 \frac{(\nu-\mu)_q(\nu)_{J-q}(J-2q+\mu)}
 {\mu\,q!(\mu+1)_{J-q}}\geq0.
 \label{eqS:d}
\end{equation}
\end{lemma}

\begin{proof}
Use the mutually inverse finite expansions
\begin{align}
 C_J^\rho(x)&=\sum_{r=0}^{\lfloor J/2\rfloor}
 (-1)^r\frac{(\rho)_{J-r}}{r!(J-2r)!}(2x)^{J-2r},
 \label{eqS:Gegen-monomial}\\
 (2x)^N&=N!\sum_{s=0}^{\lfloor N/2\rfloor}
 \frac{N-2s+\mu}{s!(\mu)_{N-s+1}}C_{N-2s}^\mu(x).
 \label{eqS:monomial-Gegen}
\end{align}
The second follows by comparing $[a^N]$ in
Eq.~\eqref{eqS:plane-wave}.  Insert Eq.~\eqref{eqS:monomial-Gegen} into
Eq.~\eqref{eqS:Gegen-monomial}.  The coefficient of $C_{J-2q}^\mu$ is
$(J-2q+\mu)T_{J,q}$, where
\begin{equation}
 T_{J,q}=\sum_{r=0}^{q}
 \frac{(-1)^r(\nu)_{J-r}}
 {r!(q-r)!(\mu)_{J-q-r+1}}.
\end{equation}
After setting $s=q-r$ and factoring terms independent of $s$,
\begin{equation}
 T_{J,q}=\frac{(-1)^q(\nu)_{J-q}}
 {q!(\mu)_{J-2q+1}}
 {}_2F_1\!\left(
 \begin{matrix}-q,\ \nu+J-q\\ \mu+J-2q+1\end{matrix};1\right).
\end{equation}
Chu--Vandermonde,
${}_2F_1(-q,a;c;1)=(c-a)_q/(c)_q$, gives
$(c-a)_q=(-1)^q(\nu-\mu)_q$ and therefore
\begin{equation}
 T_{J,q}=\frac{(\nu-\mu)_q(\nu)_{J-q}}
 {\mu\,q!(\mu+1)_{J-q}}.
\end{equation}
Multiplication by $J-2q+\mu$ proves Eq.~\eqref{eqS:d}.  Every displayed
factor has the asserted sign when $0<\mu\leq\nu$.
\end{proof}

Take $\nu=23/2$ and $\mu=(D-3)/2$.  Substitute
Eq.~\eqref{eqS:connection} into the positive $D=26$ expansion.  Every
coefficient in parameter $\mu$ is a finite nonnegative linear combination
of $D=26$ coefficients.  Hence
\begin{equation}
 B^D_{n,j}\geq0\qquad(3<D\leq26).
\end{equation}

Finally, at $n=1$,
\begin{equation}
 R_1(x)=\frac{25x^2-1}{8}.
\end{equation}
Since, for $\lambda=(D-3)/2$,
\begin{equation}
 C_2^\lambda(x)=2\lambda(\lambda+1)x^2-\lambda,
\end{equation}
one obtains
\begin{equation}
 R_1(x)=
 \frac{25}{4(D-3)(D-1)}C_2^{(D-3)/2}(x)
 +\frac{26-D}{8(D-1)}C_0(x).
\end{equation}
Thus $B^D_{1,0}<0$ for every $D>26$, proving sharpness.

\end{document}